\documentclass[envcountsame]{llncs}

 \newtheorem{obs}{Observation}

\def \qedbox{\hfill\vbox{\hrule\hbox{\vrule
 height1.3ex\hskip0.8ex\vrule}\hrule}}

\usepackage{amsmath}
  \usepackage{epsfig}

  \renewcommand*{\Pr}{\mathop{\mathrm{Prob}}}

  \def \qedbox{\hfill\vbox{\hrule\hbox{\vrule height1.3ex\hskip0.8ex\vrule}\hrule}}
  \newfont{\MSB}{msbm10 scaled\magstep1}
  
  \newcommand{\goesto}{\rightarrow}

\begin{document}
\title{Adversarial Scheduling Analysis of Game-Theoretic Models of Norm Diffusion}

\author{Gabriel Istrate\inst{1}\thanks{corresponding author. Email: gabrielistrate@acm.org}\and Madhav V.~Marathe \inst{2} \and S.S.~Ravi\inst{3}}
\institute{e-Austria Institute, V.P\^{a}rvan 4, cam. 045B, Timi\c{s}oara RO-300223, Romania \and Network Dynamics and Simulation Science Laboratory,
Virginia Tech, 1880 Pratt Drive Building XV, Blacksburg, VA 24061. Email: \email{mmarathe@vbi.vt.edu} \and Computer Science Dept., S.U.N.Y. at Albany,
  Albany, NY 12222, U.S.A. Email:\email{ravi@cs.albany.edu}}


 \maketitle

\begin{abstract}
In \cite{adversarial-soda} we advocated the investigation of robustness of results in the theory
of learning in games under adversarial scheduling models. We provide
evidence that such an analysis is feasible and can lead to nontrivial
results by investigating, in an adversarial scheduling setting, Peyton
Young's model of diffusion of norms \cite{peyton-young-book}. In 
particular, our main result incorporates {\em contagion} into Peyton Young's  model.  
\end{abstract} 


  \section{Introduction}\label{s:intro}
Game-theoretic equilibria are {\em steady-state properties};
that is, given that all the players' actions correspond to an
equilibrium point it would be irrational for any of them to deviate
from this behavior when the others stick to their strategy. The fundamental 
problem facing this type of concept is that it does not predict
{\em how players arrive at this equilibrium in the first place}, or
how they ``choose'' one such equilibrium, if several such points
exist.  The theory of equilibrium selection of 
Hars\'{a}nyi and Selten \cite{harsanyi:selten} 
assumes some form of prior coordination between players, in the form of
a {\em tracing procedure}. This strong prerequisite is often unrealistic. 

The theory of {\em learning in games} \cite{FL99} attempts to explain
the emergence of equilibria as the result of an evolutionary
``learning'' process. Models of this type assume one (or several)
populations of {\em agents}, that interact by playing a certain
game, and updating their behavior based on the outcome of this
interaction. 

Results in evolutionary game theory are important not necessarily as
realistic models of strategic behavior. Rather, they provide
possible explanations for experimentally observed features of
real-world social dynamics. For instance, the fundamental insight behind the
concept of {\em stochastically stable strategies} is that continuous
``noise'' (or small deviations from rationality) can provide a
solution to the equilibrium selection problem in game theory. 
discussion on the role of strategic learning in equilibrium selection
see \cite{peyton-young-strategic})
Similar issues apply when
mathematical modeling is replaced with computer experiments, 
in the
area of {\em agent-based social simulation} \cite{gilbert-troizch}. Epstein
\cite{epstein-generative-book} (see also \cite{axtell-epstein}) has
advocated a generative approach to social science: in order to better
understand a given phenomenon one should be able to generate it via
simulations.

Given that such mathematical models or simulations are emerging  as tools for 
policy-making (see e.g. \cite{transims-iccs,epstein-smallpox}), how can
we be sure that the conclusions that we derive from the output of the
simulation do not crucially depend on the particular assumptions and features we
embed in it ? Part of the answer is that these results have to display ``robustness'' with respect to the various idealizations inherent in the model, be it
mathematical or computational.

Various issues that might impact the robustness of the conclusions
have been previously considered in the game-theoretic literature; for
instance, the celebrated result of Foster and Young
\cite{foster:young:90} can be viewed as investigating the robustness
of Nash equilibria with respect to the introduction of small amounts
of random noise (or player mistakes). 

In this paper we are only
concerned with one such issue: {\em scheduling}, i.e.the order in
which agents get to update their strategies. Two alternatives are most
popular, both in the mathematical and the computer simulation
literature: in the {\em synchronous mode} ({\em every} player updates
at every step. A popular alternative is {\em uniform matching}. Models of
the latter type assume an underlying (hyper)graph topology (describing
the sets of players allowed to simultaneously update in one step as a
result of game playing) and choose a (hyper)edge uniformly at random
from the available ones. Employing a uniform matching model in multiagent models of social
systems is unrealistic for it assumes {\em perfect} and {\em global}
randomness; it is not clear whether this assumption is waranted in the
``real life'' situations that the theory is supposed to model.
Indeed, notwithstanding the question regarding the existence of
computational randomness in nature, the structure of social
interactions is neither random, nor uniform, and comprises many
regular, ``day by day'' interactions, as well as a smaller number of
``occasional'' ones. 
A random matching model does not take into
account {\em locality} and cannot, therefore, adequately model
``contagion'' effects (i.e. players becoming activated as a result of
some of their neighbors doing so)
. On the other hand, social systems
are {\em inherently distributed}, and it is not clear whether the
assumption of {\em global} randomness is reasonable in a simulation
setting. 

We investigate in an adversarial setting Peyton Young's model of evolution of norms
\cite{peyton-young-book} (see also \cite{peyton-young-inovations}). The dynamics
 models an important aspect of social networks, the emergence of {\em
conventions}, and has been proposed as an evolutionary justification
for the emergence of certain rules in the pragmatics of natural
language \cite{vanrooy:horn}. Our results can be summarized as follows: results on selection of
strict-dominant equilibria under random noise extend
(Theorem~\ref{adversarial-young-nonadaptive}) to a class of
nonadaptive schedulers. However, such an extension fails
 for adaptive schedulers,
even those with fairness properties similar to those of a random
scheduler. Our main result
(Theorem~\ref{young-contagion}) extends the convergence to the
strictly-dominant equilibrium to a class of ``nonmalicious'' adaptive schedulers
that models {\em contagion} and has a certain {\em
reversibility} property (the class of such schedulers includes the
random scheduler as a special case). However
 for this class of schedulers the {\em
convergence time} is {\em not} necessarily the one from the case of
random scheduling.

Besides the relevance of our results to evolutionary game-theory, we hope that the concepts and techniques relevant to this paper can be fruitfully exploited in the theory of {\em rapidly mixing Markov chains}, of great interest in Theoretical Computer Science. 

Not surprisingly, our
framework is related to the {\em theory of self-stabilization
of distributed systems} \cite{dolev:book}.  Our proofs highlight some
principles and techniques of this theory ({\em the existence of a winning
strategy for scheduler-luck games} \cite{dolev:book}, {\em
monotonicity} and {\em composition of winning strategies}) that can
be applied to the particular problem we study, and conceivably in more general
settings as well.

\section{Preliminaries}

A general class of models for which adversarial analysis can be naturally considered is that of
{\em population games} \cite{blume-population-games}. Systems
of interest in this class consist of a number of {\em agents}, defined as the
vertices of a hypergraph $H=(V,E)$. One edge of this hypergraph
represents a particular choice of all players who can play (one or
more simultaneous instances of) a game $G$ that defines the
dynamics. Each player has a {\em state} (generally a mixed strategy of
$G$) chosen from a certain set $S$. The global state of the system is an element of $\overline{S}=S^{V}$. The dynamics proceeds by choosing
one edge $e$ of $H$ (according to a scheduling mechanism that is 
specified by the scheduler), letting the agents in $e$ play the game, and
updating their states as a result of game playing.

\subsection{Schedulers}

Denote by $X^{*}$ the set of finite words over alphabet $X$.

\begin{definition}
A {\em deterministic scheduler} is specified by a mapping
$f:E^{*}\times \overline{S} \goesto E$,  where $E$ is the set of edges of
$H$ and $\overline{S}$ is the set of possible states of the system.
Mapping $f$ specifies the next scheduled element, given the 
current history. Let $b\geq 1$. A scheduler that can choose one item
  among a set of $m$ elements is {\em (worst-case)
  $b$-fair} if every agent is guaranteed to be scheduled at least once
in any sequence of $b(m-1)+1$ consecutive steps.
\end{definition}

One particularly restricted class of schedulers is that of {\em
non-adaptive} schedulers, corresponding to updates of the nodes/edges
according to a fixed permutation, independent of the initial state of
the system.

  The above definitions are well-suited for {\em deterministic
  schedulers}.  They are {\em not} well-suited for probabilistic
  schedulers (such as random matching), since for any fixed number of
  steps $B$ with positive probability it will take more than $B$ steps
  to schedule each element at least once. Also, the definitions do not
  allow for multiple agents to be scheduled simultaneously. Therefore
  in this case we need to employ slightly different definitions.
 
  \begin{definition}
  A {\em (probabilistic) scheduler} assigns a probability distribution
  $p_{w,s}$ on $E$ to each pair $(w,s,s_{0})$ consisting of initial prefixes
  $w\in E^{*}, s\in \overline{S}^{*}$  with $|w|=|s|$ and starting state $s_{0}\in \overline{S}$. The next element $e\in E$ to be
  scheduled, given prefixes $w,s$ and initial state $s_{0}$, is sampled from $E$ according to $p_{w,s,s_{0}}$.

  A {\em non-adaptive} probabilistic scheduler is specified by (a)
 a collection (multiset) $\Sigma=\{{\cal D}_{1}, \ldots, {\cal
D}_{m}\}$ of probability distributions on the set $E$ such that every
$x\in E$ belongs to the support of some distribution ${\cal D}_{i}$
and (b) a fixed permutation $\pi$ of $\Sigma$. The scheduler proceeds by (possibly concurrently) scheduling elements
of $E$ sampled from a distribution from $\Sigma$ chosen according to
(consecutive repetitions of) permutation $\pi$. For $C>0$, a non-adaptive probabilistic scheduler is {\em $C$
individually-fair} if for every $x\in E$, the probability that $x$ is
scheduled during one round of $\pi$ is at least $C/|E|$.
\end{definition}

 One can define, for any given triple $(w,s,s_{0})$, where $w\in E^{*},s\in \overline{S}^{*}$ and
 $s_{0}\in \overline{S}$, a probability $\pi_{w,s,s_{0}}$, the probability that, starting
 from state $s_{0}$ the scheduler uses $w$ as the initial prefix of its
 schedule and evolves its global state according to string $s$. Let $\Omega$ denote the resulting probability space.  We divide each
  trajectory of a probabilistic scheduler into {\em rounds}: the first
  round is the smallest initial segment that schedules each element of
  $E$ at least once, the second round is the smallest segment starting
  at the end of the first round that schedules each element at least
  once, and so on.  Given this convention, it is easy to see that for
  any $k>0$ and $s\in \overline{S}$ the family $W_{k}$ of strings $w$ consisting
  of {\em exactly $k$ rounds} realizes a complete partition of the
  probability space $\Omega$, i.e. $\sum_{w\in W_{k}} \pi_{w,s} = 1$. 

  \begin{definition}
  If $f(\cdot)$ is a function on integers, we say that a family of
  probabilistic schedulers, indexed by $n$, the cardinality of the set
  $E$, is {\em $O(f(n))$-fair w.h.p.} if there exists a monotonically
  decreasing function $g:(0,\infty)\goesto (0,1)$, with
  $\lim_{\epsilon \goesto \infty} g(\epsilon)=0$ such that for every
  state $s\in \overline{S}$, denoting by $l_{i}$ the random variable measuring
  the length of the $i$'th round, we have $\underline{\lim}_{n\goesto \infty} \Pr[ l_{i} > \epsilon \cdot f(n)]< g(\epsilon)$. 
  \end{definition}

Random scheduling can be specified by a non-adaptive probabilistic
 scheduler whose set $\Sigma$ consists of just one distribution,
 namely the uniform distribution on $E$.  This scheduler is
 1-individually fair and, by the well-known Coupon Collector Lemma
 it is also $O(n\log(n))$-fair w.h.p.

\subsection{Peyton Young's model of norm diffusion}

 The setup of this model is the following: agents
located at the vertices of a graph $G$ interact by playing a
two-person symmetric game with payoff matrix $M=(m_{i,j})_{i,j\in
\{{\bf A},{\bf B}\}}$ displayed in Figure~\ref{matrix}. It is assumed
that strategy {\bf A} is a so called {\em strict risk-dominant
equilibrium}. 
That is, we have $a-d>b-c>0$. Each undirected edge $\{i,j\}$ has a positive weight
$w_{ij}=w_{ji}$ that measures its ``importance''. When scheduled,
agents play (using the same strategy, identified as {\em the agent's
state}) against each of their neighbors. If agent $i$ is the one to update,
$\overline{x}$ is the joint profile of agents' strategies, and $z\in
\{ {\bf A}, {\bf B}\}$ is the candidate new state, $
p^{\beta}(x_{i}\goesto z | \overline{x}) \sim e^{\beta \cdot
\nu_{i}(z,\overline{x}_{-i})}$, where $\nu_{i}(z,\overline{x}_{-i} )$, the payoff of the $i$'th agent should he play
strategy $z$ while the others' profile remains the same is given by
$
\nu_{i}(z, \overline{x}_{-i}) = \sum_{(i,j)\in E} w_{ij}m_{z,x_{j}}$. 
Under random scheduling, the process we defined is a variant of the best-response dynamics. This latter process (viewed as a Markov chain) 
is not ergodic. Indeed, the since in game $G$ it is always better to play the same strategy as your partner, the dynamics has at least two fixed points,  states ``all {\bf A}" and ``all {\bf B}". 

\begin{figure}
\begin{center}
{\small
\begin{tabular}{||c|c|c||}\hline 
strategies & {\bf A} & {\bf B} \\ \hline  {\bf A} & a,a & c,d \\
\hline {\bf B} & d,c & b,b \\ \hline 
\end{tabular}
}
\caption{Payoff matrix}
\label{matrix}
\end{center}
\end{figure}

An important property of Peyton Young's dynamics is that it corresponds to a {\em potential game:} there exists a function $\rho:\overline{V}\goesto {\bf R}$ such that, for any player $i$, 
any possible actions $a_{1},a_{2}$ of player $i$, and any action profile $\overline{a}$ of the other players, $u_{i}(a_{1},a)-u_{i}(a_{2},a)= \rho(a_{1},a)-\rho(a_{2},a)$ (where $u_{i}$ is the utility function of player $i$). In other words changes in utility as a result of strategy update correspond to changes in a global potential function. An explicit potential is given by $\rho^{*}(x)=\sum_{(h,k)\in E} w_{h,k}m_{x_{h},x_{k}}$. 

\subsection{Stochastic stability} 
A fundamental concept we are dealing with is that of 
a {\em stochastically stable state} for
dynamics described by a Markov chain.

\begin{definition}\label{def:perturbed-markov}
Consider a Markov process $P^{0}$ defined on a finite state space $\Omega$.
For each $\epsilon > 0$, define a Markov process $P^{\epsilon}$ on $\Omega$.
$P^{\epsilon}$ is a {\em regular perturbed Markov process} if all
of the following conditions hold.
\begin{itemize}
\item $P^{\epsilon}$ is irreducible for every $\epsilon > 0$.
\item For every $x,y\in \Omega$, $
\lim_{\epsilon > 0} P_{xy}^{\epsilon} = P_{xy}^{0}$. 
\item If $P_{xy}>0$ then there exists $r(m)>0$, the {\em resistance of
transition $m=(x\goesto y)$}, such that as $\epsilon \goesto 0$,
$P_{xy}^{\epsilon}= \Theta(\epsilon^{r(m)})$.
\end{itemize}

Let $\mu^{\epsilon}$ be the (unique) stationary distribution of
$P^{\epsilon}$. A state $s$ is {\em stochastically stable } if $\underline{\lim}_{\epsilon \goesto 0}\mbox{  } \mu^{\epsilon}(s) > 0$.
\end{definition}

Peyton Young's model of diffusion of norms can be easily recast into the 
framework of Definition~\ref{def:perturbed-markov}, by defining  $\epsilon = exp(-\beta)$.

\begin{obs} 
Peyton Young's model of diffusion of norms can be recast into the 
framework of Definition~\ref{def:perturbed-markov}. 
Let $\epsilon = exp(-\beta)$. As $\beta \goesto \infty$,
$\epsilon \goesto 0$. Consider now the Markov chain $\Gamma_{\epsilon}$
corresponding to
the original dynamics. It has transition matrix
$D_{\epsilon}=D_{1,\epsilon}\ldots D_{m,\epsilon}$, where
$D_{i,\epsilon}=(d_{i,k,l}^{\epsilon})$ is the transition matrix
corresponding to scheduling (and updating) a node according to
${\cal D}_{i}$. It is easy to see that $\lim_{\epsilon \goesto 0} d_{k,l}^{\epsilon} = d_{k,l}$. Moreover, by the nature of the dynamics, as $\epsilon \goesto 0$
each element of $D_{i,\epsilon}$ is either zero (in case the state
transition $k\goesto l$ cannot be realized by updating any single
node member of $supp({\cal D}_{i})$), tends to a positive
constant which is the probability that the node corresponding to the
transition $k\goesto l$ is chosen (in case the transition
$k\goesto l$ corresponds to a ``best reply" move), or tends to
zero, asymptotically like $\Theta(\epsilon^{r_{i,k,l}})$, for some
$r_{i,k,l}>0$ (otherwise).
\end{obs}

\begin{definition}
A {\em tree rooted at node $j$} is a set $T$ of edges such that for
any state $w\neq j$ there exists a unique (directed) path from $w$ to
$j$. The {\em resistance of a rooted tree $T$} is the sum of resistances
of all edges in $T$.
\end{definition} 

The following characterization of stochastically stable states 
is presented as Lemma 3.2 in the Appendix of \cite{peyton-young-book}: 

\begin{proposition}\label{characterization}
Let $P^{\epsilon}$ be a regular perturbed Markov process, and for each
$\epsilon > 0$ let $\mu^{\epsilon}$ be the unique stationary distribution
of $P^{\epsilon}$. Then $\lim_{\epsilon \goesto 0} \mu^{\epsilon}
= \mu^{0}$ exists, and $\mu^{0}$ is a stationary distribution of
$P^{0}$. The stochastically stable states are precisely those states
$z$ such that there exists a tree rooted at $z$ of minimal resistance
(among all rooted trees).
\end{proposition}


\begin{definition}
Given a graph $G$, a nonempty subset $S$ of vertices and a real number $0\leq r\leq 1/2$ we say that {\em $S$ is $r$-close-knit} if $
\forall S^{\prime}\subseteq S, S^{\prime}\neq \emptyset,~~ 
\frac{e(S^{\prime},S)}{\sum_{i\in S^{\prime}} deg(i)} \geq r$, 
where $e(S^{\prime},S)$ is the number of edges with one endpoint 
in $S^{\prime}$ and the other in $S$, and $deg(i)$ is the degree of vertex $i$. A graph $G$ is {\em $(r,k)$-close-knit} if every vertex 
is part of a $r$-close-knit set $S$, with $|S|=k$. 
\end{definition}

\begin{definition}
Given $p\in [0,1]$, the {\em $p$-inertia} of the process is the maximum,
over all states $x_{0}\in S$, of $W(\beta,p,x_{0})$,
the expected waiting time until at least $1-p$ of the population
is playing action $A$ conditional on starting in state $x_{0}$.
\end{definition}

The model in \cite{peyton-young-inovations,peyton-young-book} assumes independent individual
updates, arriving at random times governed (for each agent) by a Poisson arrival process with rate one. Since we are, however, interested in
adversarial models that do not have an easy description in continuous time
we will assume that the process proceeds in discrete steps. At each such step
a random node is scheduled. It is a simple exercise 
to translate the result in \cite{peyton-young-inovations,peyton-young-book}
to an equivalent one for global, discrete-time scheduling. The conclusions of this translation are: 
\begin{itemize} 
\item The stationary distribution of the process is {\em the Gibbs distribution}, $\mu_{\beta}(x)=\frac{e^{\beta\rho(x)}}{\sum_{z} e^{\beta\rho(x)}}$, where $\rho$ is the potential function of the dynamics.  
 \item ''All {\bf A}`` is the unique stochastically-stable state of the dynamics. 
\item Let $r^{*}=\frac{b-c}{a-d+b-c}$, and let $r>r^{*}$, $k>0$. 
On a family of $(r,k)$-close-knit graphs the convergence time is $O(n)$. 
\end{itemize}

\section{Results} 

First we note that Peyton Young's results easily extend to non-adaptive schedulers.  \emph{Adaptive} schedulers on the other hand, even those of fairness no
higher than that of the random scheduler, can preclude the system from ever entering a state where a proportion higher than $r$ of agents plays the risk-dominant strategy:

\begin{theorem}\label{adversarial-young-nonadaptive}
The following hold:  \\ 

(i) For {\em all} non-adaptive schedulers, the state ``all {\bf A}'' 
is the unique stochastically stable state of the system. 

(ii) Let ${\cal G}$ be a class of graphs that are $(r,k)$-close-knit for
some fixed $r>r^{*}$. Let $f=f(n)$ be a class of non-adaptive $\Theta(1)$
individually fair schedulers. Given any $p\in (0,1)$ there exists a
$\beta_{p}$ such that for all $\beta > \beta_{p}$ there exists a constant
$C$ such that the $p$-inertia of the process (under scheduling given
by $f$) is at most $C\cdot m\cdot n$, where $m=m(n)$ is the number of
rounds of $f$ and  $n$ is the number of vertices of the underlying graph.

(iii) For every $0<r<1$ there exists an {\emph{adaptive}} scheduler which is
$O(n\log(n))$-fair w.h.p. (where the constant hidden in the ``O'' notation
depends on $r$) that can forever prevent
the system, started on the ``all {\bf B}'s'' configuration, from ever 
having more than a fraction of $r$ of the agents playing 
${\bf A}$.
\end{theorem}

\begin{enumerate}
\item
 Let
$i\in V$ and let $M_{i}$ be the restriction of the given dynamics
corresponding to the case when only one node, node $i$, is
scheduled in all moves (otherwise the dynamics is similar to the
original one).

It is easy to see that $M_{i}$ is a non-ergodic Markov chain and that $\mu_{\beta}$ is {\em a}
stationary distribution for the Markov chain $M_{i}$. This is
so because for two configurations $x,y$ that only differ in position $i$,
the ratio of transition probabilities $p_{x,y}/p_{y,x}$ is equal to
to $exp[\beta \cdot (\rho^{*}(x)-\rho^{*}(y))]$, which is precisely $\mu_{\beta}(x)/\mu_{\beta}(y)$.

Now consider the matrix $D_{k}$ corresponding to the distribution
with the same notation as in the periodic schedule. It is a convex
combination of the matrices $M_{i}$, hence it will also have
$\mu_{\beta}$ as a stationary distribution. We infer that the
product of matrices $D_{k}$ corresponding to the cyclic schedule
also has $\mu_{\beta}$ as a stationary distribution.

But it is easy to see that the Markov chain corresponding to one
round of the cyclic schedule is irreducible (since one can
navigate between any two states in at most $|V|$ rounds, by
flipping the differing bits and keeping the bits that coincide
fixed) and aperiodic (since the probability of remaining in a
given state is positive). Therefore, it must have an unique
stationary distribution, which is necessarily $\mu_{\beta}$.

\item

Consider a vertex $v\in V$ and a $r$-close-knit set of size $k$ containing $v$,
denoted $S_{v}$. Consider $\Gamma_{v,\beta}$ the version of the process where
each vertex in $S_{v}$ updates just as before, but each vertex in $V\setminus S_{v}$
always chooses state $B$ when scheduled.

This restricted dynamics on $V$ still corresponds to a potential game,
specified by potential function
\[
\rho^{*}(x) = \sum_{(i,j)\in E} \rho(x_{i}, x_{j}),\mbox{ for }x\in \{A,B\}^{S_{v}}B^{V\setminus S_{v}},
\]

and with 
the Gibbs distribution 
$\mu^{\beta}(x)= \frac{ e^{\beta \cdot \rho^{*}(x)}}{\sum_{z}
e^{\beta \cdot \rho^{*}(z)} }$ as its stationary distribution. 
Again, just as in \cite{peyton:young:conventions} (since the precise scheduling
order does {\em not} play a role in this result) the condition that $G$ is
$(r,k)$-close-knit implies that the state $A_{S}$ defined as ``all {\bf A}'' on $S_{v}$ and ``all {\bf B}'' on $V\setminus S_{v}$  is the state with the
highest potential among the possible states of the system.

One additional complication of the dynamics $\Gamma_{v,\beta}$ is that
it often schedules (unnecessarily) nodes outside $S_{v}$, that do not change. Consider
$\Xi_{v,\beta}$ that is the version of $\Gamma_{v,\beta}$ that ``only schedules nodes in $S_{v}$'' (i.e. it ignores moves of $\Gamma_{v,\beta}$ that schedule nodes outside of $S_{v}$).

To describe this dynamics formally, view each distribution $D_{i}$
as a set of symbols from the alphabet $V$. Then the set of
trajectories of the dynamics $\Gamma_{v,\beta}$ can be specified
by the words of the regular language $L_{\Gamma} = (D_{1}\cdot D_{2}\cdot \ldots
\cdot D_{m})^{*}$. Trajectories of $\Xi_{v,\beta}$ correspond to words
in another regular language $L_{\Xi}$, more precisely to the ones
corresponding to deleting symbols in $V\setminus S_{v}$ from words
in $L_{\Gamma}$. This regular language can be specified by the
regular expression  $((D_{1}\cup\{\lambda\})\cdot \ldots \cdot (D_{m}\cup
\{\lambda\}) \cap S_{v}^{+})^{*}$. This expression yields a matrix
 of size $2^{|S_{v}|}\times 2^{|S_{v}|}$ for $\Xi_{v,\beta}$.

\vspace{5mm}
\begin{claim}\label{firstclaim} 
For every $\epsilon > 0$ there exists $\eta \in {\bf N}$ such that, for every $\eta^{\prime}> \eta \in {\bf N}$ any initial state of $\Gamma_{v,\beta}$ and every state
$T\in \{{\bf A,B}\}^{S_{v}}$.
\[
|Pr[\Gamma_{v,\beta}\mbox { in state }T\mbox{ } |\mbox{ } |w| = \eta^{\prime} \cdot m\cdot n]-\Pi(T)| \leq \epsilon.
\]

\end{claim}

\begin{proof}

Let $\epsilon > 0$. As $\Xi_{v,\beta}$ converges to its stationary
distribution, there exists $\overline{k}>0$ such that $\forall k^{\prime} > \overline{k}$ and every initial state of $\Xi_{v,\beta}$ 

\begin{equation}\label{approx-2}
|Pr[\Xi_{v,\beta}\mbox { in state }T | |w| = k^{\prime}] - \tilde{\Pi}(T)| \leq \epsilon/2,
\end{equation} 

where $\tilde{\Pi}$ is the stationary distribution of dynamics $\Xi$. Of course, states with positive support in $\tilde{\Pi}$ have the same probability
in $\Pi$, that is

\[
\forall T \in \{{\bf A,B}\}^{S_{v}}\mbox{: } \Pi(T) = \tilde{\Pi}(v).
\]

Let $Y$ be a random trajectory of length $\eta^{\prime}\cdot m\cdot n$ in $l_{\Gamma}$ and let $pr(Y)$ its projection onto $L_{\Xi}$.

\begin{claim}\label{approx} 
There exists $\eta > 0$ such that $\forall \eta^{\prime} > \eta$
\[
\Pr_{|Y|= \eta^{\prime}\cdot m \cdot n}[|pr(Y)| < \overline{k}] \leq \frac{\epsilon}{2}.
\]
\end{claim}

\begin{proof}
The probability that any
given distribution $D_{i}$ whose support includes some element in $S_{v}$
will schedule (in a given round) a node in this set is $\Omega(1/n)$, by the fairness condition. There is at least one such $D_{i}$ among all the $m$ distributions. Therefore, the expected length of $pr(Y)$ is $\Omega(k/n)$.  A simple application of Markov's inequality gives the desired result.
\end{proof}
\qedbox

Now write
\begin{eqnarray*}
& Pr & [\Gamma_{v,\beta}\mbox { in state }T | |w|  = \eta^{\prime} \cdot m\cdot n] \\  = & \sum_{j} & Pr[\Gamma_{v,\beta}\mbox { in state }T \mbox{ }|\mbox{ } |w| = \eta^{\prime} \cdot m\cdot n, |pr(w)|=j]\\ & \cdot & \Pr[|pr(w)|=j] \\ 
\end{eqnarray*}

Therefore we have 
\begin{eqnarray*}
& Pr & [\Gamma_{v,\beta}\mbox { in state }T | |w|  =  \eta^{\prime} \cdot m\cdot n] -\Pi (t)| \leq \\  \leq  & \sum_{j} & |Pr[\Gamma_{v,\beta}\mbox { in state }T \mbox{ }|\mbox{ } |w| = \eta^{\prime} \cdot m\cdot n, |pr(w)|=j]-\Pi(T)|\cdot \\ & \cdot &  \Pr[|pr(w)|=j] 
\end{eqnarray*}

The first term in the product is an absolute difference between
two probability values, and thus 
has absolute value at most one. Therefore, by Claim~\ref{approx}, if we neglect in the sum those terms with $j<k$ only changes the sum by at most $\epsilon/2$. On the other hand 
\begin{eqnarray*}
Pr[\Gamma_{v,\beta}\mbox { in state }T \mbox{ }|\mbox{ } |w| = \eta^{\prime} \cdot m\cdot n, |pr(w)|=j] \\
= Pr[\Xi_{v,\beta}\mbox { in state }T | |w| = k^{\prime}]. 
\end{eqnarray*}

Now, using Equation~(\ref{approx-2}),  Claim~\ref{firstclaim} follows.  
\end{proof}
\qedbox 

From now on the proof mirrors rather closely the one for the case of random scheduling (presented in the Appendix to \cite{peyton:young:conventions}): first, because the stationary distribution of process $\Gamma$ is the Gibbs distribution, there exists a finite value $\beta(\Gamma,S,p)$ such that $\mu_{v}(A_{S})\geq 1-p^{2}/2$ for all $\beta > \beta(\Gamma,S,p)$. 

There are only a finite number of nonisomorphic dynamical systems
$\Gamma_{v,\beta}$ (where isomorphism of dynamical system is meant to
be the isomorphism of the underlying graph topologies $S_{v}$ and of
the projection of schedulers onto $S_{v}$). In particular we can find
$\beta(r,k,p)$ and $\eta(r,k,p)$ such that, for all graphs $G\in {\cal
G}$ and all $r$-close-knit subsets $S$ with $k$ vertices, the following
relation holds for {\em all} initial states:

\begin{equation}
\forall \beta \geq \beta(r,k,p), \forall \eta^{\prime} \geq \eta, Pr[y^{\eta^{\prime}\cdot m \cdot n}= A_{S}] \geq 1-p^{2}.  
\end{equation}   

where $y^{t}$ is the state of the dynamical system on state $S_{v}$ at time $t$. We can now derive that for every close-knit set $S$
\begin{equation}
\forall \beta \geq \beta(r,k,p), \forall \eta^{\prime} \geq \eta, Pr[x^{\eta^{\prime}\cdot m \cdot n}= A_{S}] \geq 1-p^{2}.  
\end{equation}   

where $x_{t}$ is now the state of the process from the theorem. The
argument is obtained via essentially the same coupling as the one from
\cite{peyton-young-book}, hence it is omitted from this writeup. Since every vertex $i$
is contained in a $(r,k)$-close-knit set, it follows that

\[
\forall \beta \geq \beta(r,k,p), \forall \eta^{\prime} \geq \eta, Pr[x_{i}^{\eta^{\prime}\cdot m \cdot n}= {\bf A}] \geq 1-p^{2}.
\]

Therefore the expected proportion of vertices playing action $A$ at time $\eta^{\prime}\cdot m \cdot n$ is at least $(1-p^{2})n$. 

But this implies that 
\[
\forall t \geq \eta \cdot m \cdot n,  Pr[\mbox{ at least }(1-p)n\mbox{ nodes have label }{\bf A}] \geq 1-p. 
\]

Indeed, if this wasn't the case, then with probability at least $p$ more than 
$pn$ nodes at time $t$ would have label $B$, a contradiction. Now, by an application of Markov's inequality, the expected time until at least $(1-p)n$ nodes 
are labelled $A$ is at most $\eta\cdot m \cdot n/(1-p)$. Since this holds for 
all graphs $G$ in ${\cal G}$, the $p$-inertia of the process is bounded as stated in the theorem. 
\item 

Consider a scheduler working in rounds. In each round the scheduler is
scheduling nodes according to a fixed permutation $\pi$, the same for all rounds.
In each round the scheduler is scheduling each node at least once. For the
first $\lceil rn \rceil+1$ nodes the scheduler continues scheduling each of them
(after the initial one) until the node switches to strategy ${\bf B}$. The scheduler plays each remaining node {\em exactly} once.

It is easy to see that there exists a constant $\epsilon >0$ (that may depend on $\beta$) such that, at each stage, each agent switches to strategy ${\bf B}$
with probability greater or equal to $\epsilon$.

Therefore the probability that any given agent needs to be scheduled for
more than $c\log(n)$ rounds before turning to ${\bf B}$ is $o(1/n)$
for large enough $c$. It follows that the given scheduler is $O(n\log(n))$-fair w.h.p.
\end{enumerate}

\subsection{Main result: Diffusion of norms by contagion}

Adaptive schedulers can display two very different notions of adaptiveness: 
\begin{enumerate} 
 \item The next node depends only on the set of previously scheduled nodes, or \item It crucially depends on the {\em states} of the system so far. 
\end{enumerate}

The adaptive schedulers in Theorem~\ref{adversarial-young-nonadaptive} (iii) was crucially using the second, stronger, kind of adaptiveness. In the sequel we study a model that displays adaptiveness of type (1) but not of type (2). The model is specified as follows: To each node $v$
we associate a probability distribution 
$D_{v}$  on the vertices of $G$. We then choose the next scheduled node 
according to the following process. If $t_{i}$ is the
node scheduled at stage $i$, we chose $t_{i+1}$, the next
scheduled node, by sampling from $D_{t_{i}}$. In other words, the
scheduled node performs a (non-uniform) random walk on the vertices
of graph $G$. To exclude technical problems such as the
periodicity of this random walk, we assume that it is always the
case that $v\in supp(D_{v})$. Also, let $H$ be the directed graph with 
edges defined as follows: $(x,y)\in E[H] \iff (y \in supp(D_{x}))$. 
This dynamics generalizes both the class of non-adaptive schedulers
from previous result and the random scheduler (for the case when
$H$ is the complete graph). In the context of van Rooy's
evolutionary analysis of signalling games in natural language
\cite{vanrooy:horn}, it functions as a simplified model for an
essential aspect of emergence of linguistic conventions:
transmission via {\em contagion}.

It is easy to see that the dynamics can be described by an
aperiodic Markov chain $M$ on the set on $V^{\{{\bf A,B}\}} \times V$, where a
state $(\overline{w},x)$ is described as follows:
\begin{itemize}
\item $\overline{w}$ is the set of strategies chosen by the
agents.

\item $x$ is the label of the last agent that was given the chance to
update its state.
\end{itemize}

If the directed graph $H$ is strongly connected then the Markov chain $M$ is irreducible, hence it
has a stationary distribution $\Pi$. We will, therefore, limit ourselves in the sequel to settings with strongly connected $H$. We will, further, assume that the dynamics is {\em weakly reversible}, i.e. 
$(x\in supp(D_{y}))$ if and only if  $(y\in supp(D_{x}))$. This, of course, means that the graph $H$ is undirected. Note that since we do not constrain otherwise the transition probabilities of distributions $D_{i}$,  the stationary distribution $\Pi$ of the Markov chain
does {\em not}, in general, decompose as a product of component
distributions. That is, one cannot generally write $\Pi(w,x)$ as 
$\Pi(w,x) = \pi(w)\cdot \rho(x)$, for some distributions $\pi, \rho$.

\begin{theorem}\label{young-contagion} 
\label{rwalk} 
The set  $Q=\{(w,x)|w = V^{\bf A}\}$ is the set of stochastically stable states for the diffusion of norms by contagion.  
\end{theorem} 

\begin{proof} 
The states in $Q$ are obviously reachable from one another by
zero-resistance moves, so it is enough to consider one state $y\in Q$
and prove that it is stochastically stable. To do so, by Proposition~\ref{characterization}, all we need to do is
show that $y$ is the root of a tree of minimal resistance.
Indeed, consider another state $x\in Q$ and let $T$ be a minimum
potential tree rooted at $x$. 
\begin{claim} 
There exists a tree
$\overline{T}$ rooted at $y$ having potential less or equal to the
potential of the tree $T$, strictly smaller in case 
$x$ is not a state having
all its first-component labels equal to $A$.
\end{claim} 

Let $
\pi_{y,x}= (x_{0},i_{0})\goesto (x_{1},i_{1})\goesto \ldots \goesto (x_{k},i_{k})\goesto (x_{k+1},i_{k+1})\goesto \ldots \goesto (x_{r},i_{r})$ 
be the path from $y$ to $x$ in $T$ (that is $(x_{0},i_{0})=y$, $(x_{r},i_{r})=x$). 

We will define $\overline{T}$ by viewing the set of edges of $T$ as partitioned into subsets of edges corresponding to  
paths as follows (see Figure~\ref{tree-fig} (a)): 

\begin{enumerate} 
\item The set of edges of path $\pi_{y,x}$. 
\item The set of edges of the subtree rooted at $y$. 
\item Edges of tree components (perhaps consisting of a single node) rooted
at a node of $\pi_{y,x}$, other than $y$ (but possibly being $x$).  
\end{enumerate} 

\begin{figure} 
\begin{center}
\includegraphics*[width=4.5cm]{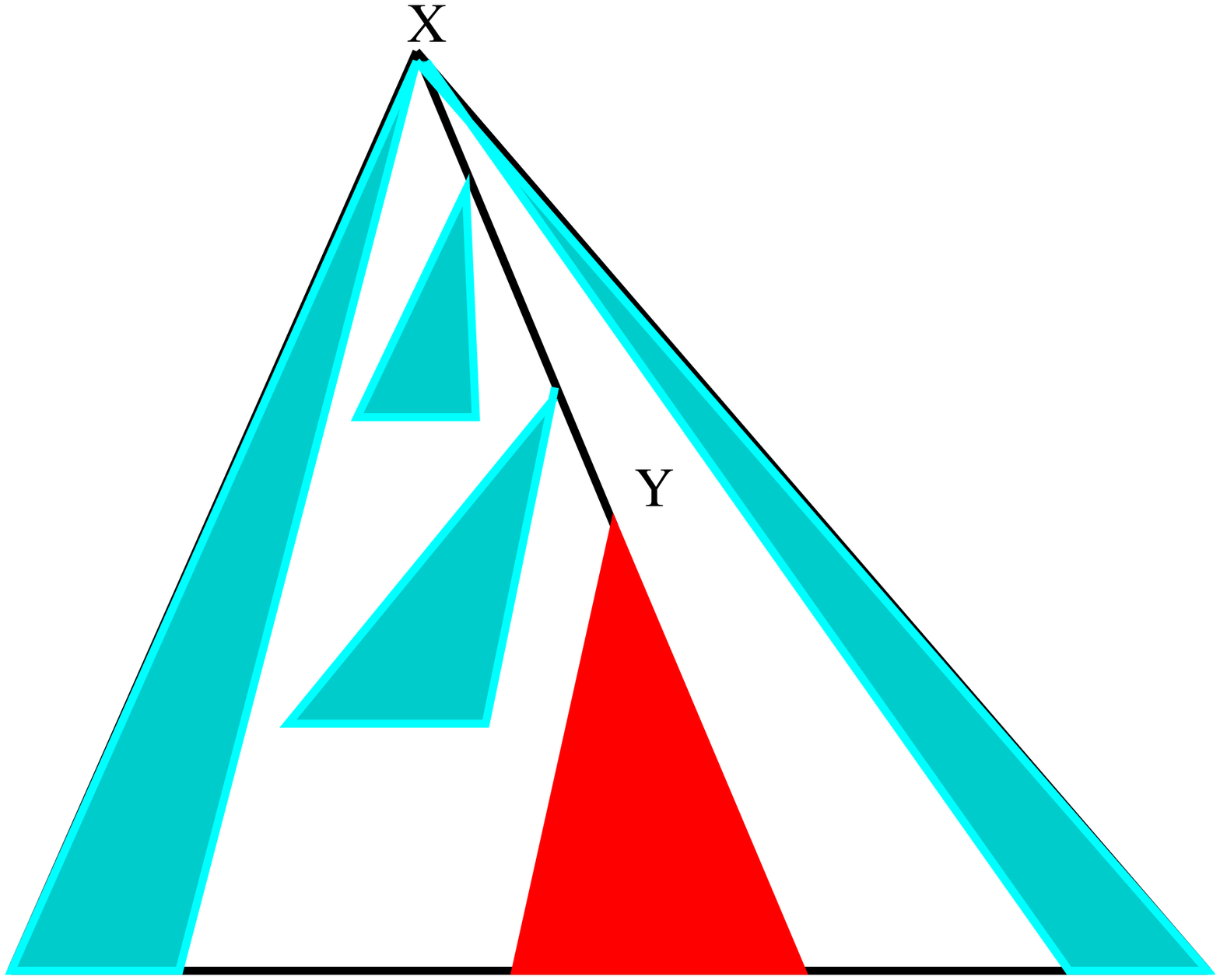}
\includegraphics*[width=4.5cm]{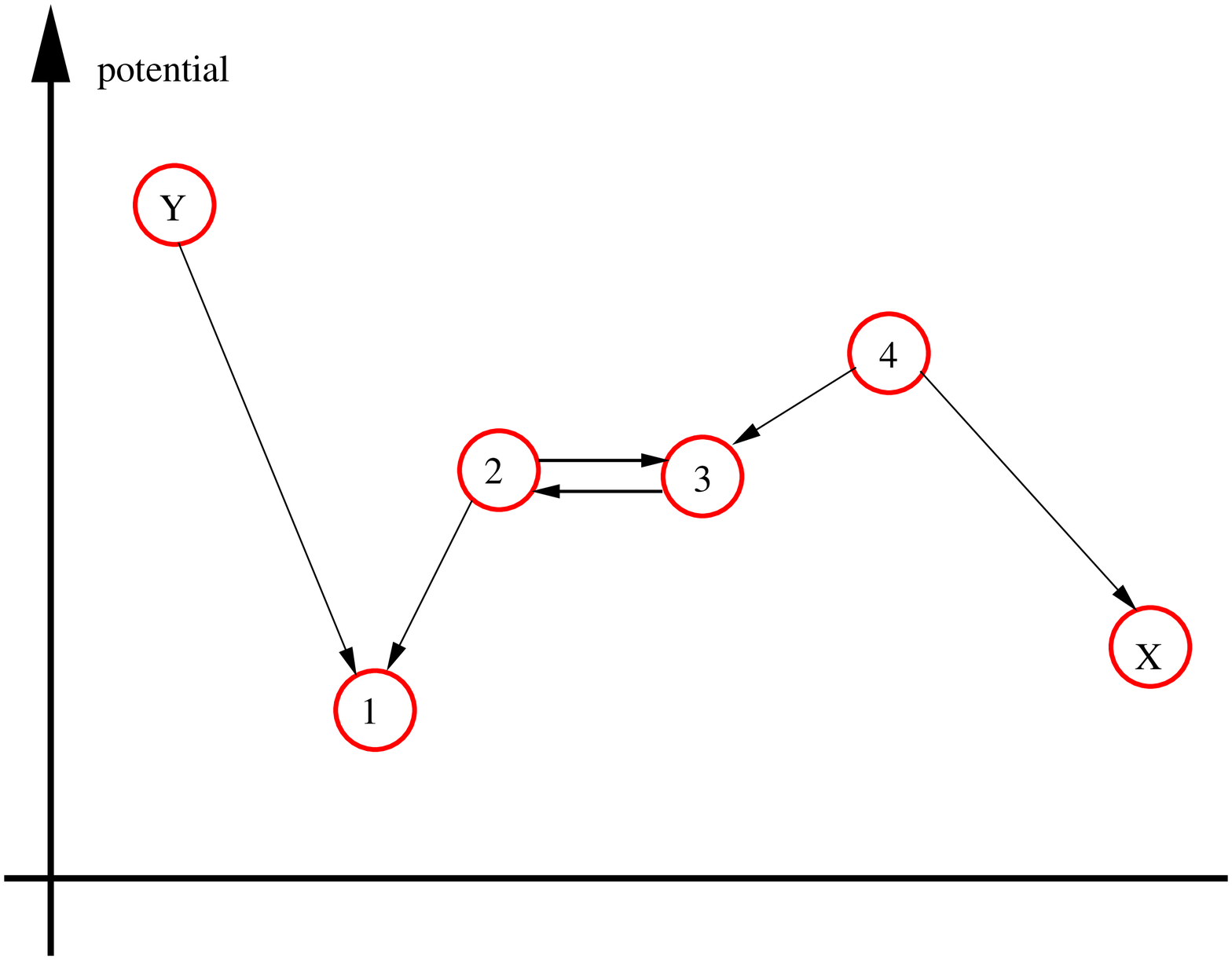}
\caption{(a). Decomposition of edges of tree $T$ (b). Resistance of edges on a path between two nodes $X$ and $Y$.}
\label{tree-fig}
\end{center} 
\end{figure}

\noindent

To obtain $\overline{T}$ we will transform each tree (path) in the above decomposition of $T$ into one that will be added to $\overline{T}$.  
The transformation goes as follows: 

\begin{enumerate}
\item Instead of path $\pi_{y,x}$ we add path $\Pi_{x,y}$ from $x$ to $y$ defined by: $\Pi_{x,y} = (x_{r},i_{r})\goesto (x_{r-1},i_{r})\goesto (x_{r-2},i_{r-1}) \goesto \ldots \goesto (x_{0},i_{1}) \goesto (x_{0},i_{0})=y$. 
\item Rooted trees of type (2) are included into tree $\overline{T}$ as well. 

\item The transformation is more complicated for the third type of edges, and we explain it in detail. Let $W_{k}$ be a tree component of $T$, connected to path $\pi_{y,x}$ at connection point $(x_{k},i_{k})$. 

{\bf Case 1:} {\em $x_{k}=x_{k-1}$}. Then the point $(x_{k},i_{k})= (x_{k-1},i_{k})$ belongs to path $\Pi_{x,y}$ as well, so one can just add the rooted tree $W_{k}$ to $\overline{T}$ as well.  

{\bf Case 2:} {\em $x_{k}\neq x_{k-1}$ and the move $(x_{k-1},i_{k-1})\goesto (x_{k},i_{k})$ has positive resistance}. In this case, since in configuration $x_{k-1}$ and scheduled node $i_{k}$ 
we have a choice of either moving to $x_{k}$ or staying in $x_{k-1}$, it 
follows that the move $(x_{k},i_{k})\goesto (x_{k-1},i_{k})$ has zero 
resistance. 

Therefore we can add to  $\overline{T}$ the tree $\overline{W_{k}}= W_{k}\cup \{(x_{k},i_{k})\goesto (x_{k-1},i_{k})\}$. The tree has the same resistance as the one of tree $W_{k}$. 

{\bf Case 3:} $x_{k-1}\neq x_{k}$ and the move $(x_{k-1},i_{k-1})\goesto (x_{k},i_{k})$ has zero resistance. 

Let $j$ be the smallest integer such that either $x_{k+j+1}=x_{k+j}$ or $x_{k+j+1}\neq x_{k+j}$ and the move $(x_{k+j},i_{k+j})\goesto (x_{k+j+1},i_{k+j+1})$ has positive resistance. 

In this case, one can first replace $W_{k}$ by $W_{k} \cup \{(x_{k},i_{k})\goesto (x_{k+1},i_{k+1})$, $(x_{k+1},i_{k+1})\goesto \ldots \goesto (x_{k+j},i_{k+j}) \}$ without increasing its total resistance. Then we apply one of the techniques from Case 1 or Case 2. 

{\bf Case 4:} $x_{k-1}\neq x_{k}$, the move $(x_{k-1},i_{k-1})\goesto (x_{k},i_{k})$ has zero resistance, and all moves on $\pi_{y,x}$, from $x_{k}$ up to $x$ have zero resistance. Then define $\overline{W_{k}}=W_{k}\cup (x_{k},i_{k})\goesto (x_{k+1},i_{k+1})\goesto \ldots \goesto x$.

\end{enumerate} 

It is easy to see that no two sets $W_{k}$ intersect on an edge having positive resistance. The union of the paths of all the sets is a
directed associated graph $W$ rooted at $y$, that contains a rooted tree $\overline{T}$ of potential no larger than the potential of $W$.  Since transformations in cases (i),(iii) do not increase tree resistance, to compare the potentials of $T$ and $W$
it is enough to compare the resistances of paths $\pi_{y,x}$ and
$\Pi_{x,y}$.

We come now to a fundamental property of the game $G$: since it is a potential 
game, the resistance $r(m)$ of a move 
$m= (a_{1},j_{1})\goesto (a_{2},j_{2})$ only 
depends on the values of the potential function at three points: 
$a_{1}, a_{2}$ and $a_{3}$, where $a_{3}$ is the state 
obtained by assigning node $j_{2}$ the value not assigned 
by move to $a_{2}$. Specifically, $r(m)>0$ if either $\rho^{*}(a_{2}) < \rho^{*}(a_{1})$, in which case 
$r(m)= \rho^{*}(a_{1})- \rho^{*}(a_{2})$, or 
\item $a_{2}=a_{1}$ and $\rho^{*}(a_{3}) > \rho^{*}(a_{1})$, in 
which case $r(m)= \rho^{*}(a_{3})- \rho^{*}(a_{1})$. In other words, the resistance of a move is positive in the following two cases:  (1)
The move leads to a decrease of the value of the potential function. In this case the resistance is equal to the difference of potentials. 
(2) The move corresponds to keeping the current state 
(thus not modifying the value of the potential function), 
but the alternate move would have increased the 
potential. In this case the resistance is equal to the value 
of this increase.  

Let us now compare the resistances of paths $\pi_{y,x}$ and
$\Pi_{x,y}$. First, the two paths contain no edges of infinite resistance,
since they correspond to possible moves under Markov chain dynamics
$P^{\epsilon}$. If we discount second components, the two paths
correspond to a single sequence of states $Z$ connecting $x_{0}$ to
$x_{r}$, more precisely to {\em traversing $Z$ in opposite directions}.
(The last move in $\Pi_{x,y}$ has zero resistance and can thus be
discounted). Resistant moves of type (2) are taken into account by both traversals, 
and contribute the same resistance value to both paths. 
So, to compare the resistances of the two paths it is enough 
to compare resistance of moves of type (1). Moves of type (1) of positive resistance are those that lead to a decrease 
in the potential function. Decreasing potential in one direction corresponds to increasing it in the other (therefore such moves have zero resistance in the opposite direction).

An illustration of the two types of moves is given in
Figure~\ref{tree-fig} (b), where the path between $X$ and $Y$ goes through
four other nodes, labeled   1 to 4. The relative height of each node
corresponds to the value of the potential function at that node. Nodes 2
and 3 have equal potential, so the transition between 2 and 3 contributes
an equal amount to the resistance of paths in both directions (which may be positive or not). Other than that only transitions of positive resistance are pictured.  

The conclusion of this argument is that $
r(\pi_{y,x})-r(\Pi_{x,y}) = \rho^{*}(x)-\rho^{*}(y) \geq 0$, 
and $r(\pi_{y,x})-r(\Pi_{x,y})> 0$ unless $x$ is an  ``all $A$'' state. 
\end{proof}

\subsection{The inertia of diffusion of norms with contagion} 

Theorem~\ref{rwalk} shows that random scheduling is not essential in ensuring that  stochastically stable states in Peyton Young's model correspond to all players playing $A$: the same result holds in the model with contagion. On the other hand, the result on the $p$-inertia of the process on families of close-knit graphs is {\em not} robust to 
such an extension. Indeed, consider the {\em line graph} $L_{2n+1}$ on $2n+1$ nodes labelled  $-n, \ldots, -1, 0, 1 \ldots n$. Consider a random walk model such that: (a) the origin of the random walk is node $0$, and (b) the walk goes left, goes right or stays in place, each with probability $1/3$. It is a well-known property of the random walk that it takes $\Omega(n^{2})$ 
time to reach nodes at distance $\Omega(n)$ from the origin. Therefore, the $p$-inertia of this random walk dynamics is $\Omega(n^{2})$ even though for every $r>0$ there exists a constant $k$ such that the family $\{L_{2n+1}\}$ is $(r,k)$-close-knit for large enough $n$. 

In the journal version of the paper we will present an upper bound on the $p$-inertia for the diffusion of norms with contagion based on concepts similar to the {\em blanket time} of a random walk \cite{blanket-time}. 

\section{Conclusions and Acknowledgments}

Our results have made the original statement by Peyton Young more robust, and have highlighted the (lack of) importance of various properties of the random scheduler in the results from \cite{peyton-young-book}: the {\em reversibility} of the random scheduler, as well as its inability to use the global system state are important in an adversarial setting, while its fairness properties are not crucial for convergence, only influencing convergence time. Also, the fact that the stationary distribution of the perturbed process is the Gibbs distribution (true for the random scheduler) does not necessarily extend to the adversarial setting. 

This work has been supported by the Romanian CNCSIS under a PN-II ``Idei'' Grant, by the U.S.\ Department of Energy under contract W-705-ENG-36 and  by NSF Grant CCR-97-34936.

\bibliographystyle{alpha}
\bibliography{/home/gistrate/bib/bibtheory}
\end{document}